\documentclass[11pt]{article}
\usepackage{amsmath}
\usepackage{amsfonts}
\usepackage{amssymb}
\usepackage{amsthm}
\usepackage[pdftex]{color,graphicx}

\usepackage{tikz-cd}
\usepackage{tkz-euclide}
\newtheorem{Theorem}{Theorem}

\newtheorem{Cor}{Corollary}

\newtheorem{Lemma}{Lemma} 

\newtheorem{Def}{Definition}

\begin{document}

% Blackboard bold letters
\def\Z{\mathbb{Z}}
\def\R{\mathbb{R}}
\def\C{\mathbb{C}}
\def\S{\mathbb{S}}
\def\T{\mathbb{T}}
\def\H{\mathbb{H}}

\newcommand{\ket}[1]{|{#1}\rangle}
\newcommand{\bra}[1]{\langle{#1}|}
\newcommand{\states}{\mathcal{S}}

\def\pP{\mathcal{P}}
\def\cC{\mathcal{C}}
\def\oO{\mathcal{O}}
\def\eE{\mathcal{E}}

\def\ss{\mathfrak{s}}

\newcommand{\set}[1]{\ensuremath{ \lbrace #1 \rbrace }}
\newcommand{\Tr}{\text{Tr}}
\newcommand{\Span}[1]{\ensuremath{ \langle #1 \rangle }}
\newcommand{\CF}{\text{CF}}
\newcommand{\CFs}{{\sf{CF}}}
\newcommand{\NCF}{\text{NCF}}
\newcommand{\NCFs}{{\sf{NCF}}}
\newcommand{\blu}[1]{{\color{magenta}{#1}}}

\newtheorem{Fact}{Fact}
\newtheorem{Definition}{Definition}

\title{The cohomological  and the resource-theoretic perspective on quantum contextuality:\\ 
common ground through the contextual fraction}

\author{$\text{Cihan Okay}^{1,2}$, $\text{Emily Tyhurst}^1$ and $\text{Robert Raussendorf}^{1,2}$\vspace{4mm}\\
{\small{\em{1: Department of Physics \& Astronomy, University of British Columbia, Vancouver, BC V6T1Z1, Canada,}}} \vspace{0.2mm}\\
{\small{\em{2: Stewart Blusson Quantum Matter Institute, University of British Columbia, Vancouver, BC, Canada}}}
}

\maketitle

\begin{abstract}
 We unify the resource-theoretic and the cohomological perspective on quantum contextuality. At the center of this unification stands the notion of the contextual fraction. For both symmetry and parity based contextuality proofs, we establish cohomological invariants which are witnesses of state-dependent contextuality. We provide two results invoking the contextual fraction, namely (i) refinements of logical contextuality inequalities, and (ii) upper bounds on the classical cost of Boolean function evaluation, given the contextual fraction of the corresponding measurement-based quantum computation. 
\end{abstract}

\section{Introduction}\label{Intro}

Contextuality \cite{KS}--\cite{CSW} is a fundamental property of quantum mechanics that distinguishes it from classical physics. The classical view of a physical system assumes that there are predefined outcomes for experiments which measurements simply reveal. Non-contextuality then means that the value corresponding to any given observable is independent of which other compatible observables might be measured simultaneously. However, it turns out that for sufficiently complex quantum systems (Hilbert space dimension $\geq 3$), no non-contextual classical model can reproduce the predictions of quantum mechanics \cite{KS},\cite{Bell}. The latter is therefore called {\em{contextual}}.

Contextuality is also important for the functioning of quantum computation. Its necessity has been demonstrated for  the models of quantum computation with magic states \cite{BK}, see \cite{How}--\cite{Lov}, and measurement-based quantum computation (MBQC) \cite{RB01}, see \cite{AB}--\cite{Abram3}. It is therefore natural to consider contextuality as a computational resource.

Of interest for the present paper is the phenomenology contained in the triangle
$$
\parbox{10cm}{\includegraphics[width=9cm]{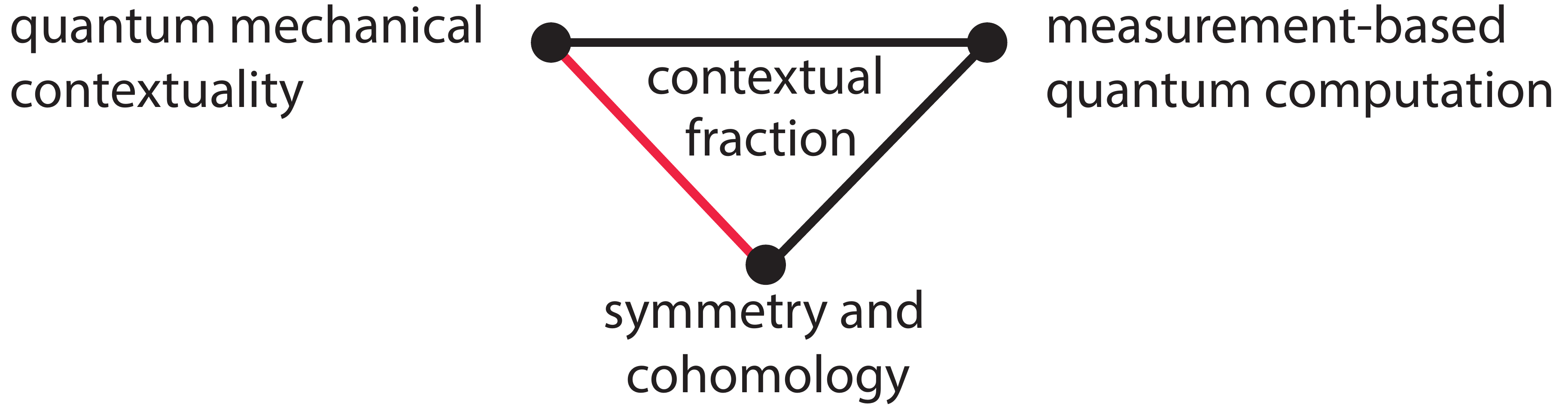}}.
$$
Therein, the connection between contextuality and MBQC (top leg) was discovered in \cite{AB}, and further studied in \cite{Hob}--\cite{Abram3}. A cohomological underpinning of contextuality, based on  \v{C}ech cohomology, was first described in \cite{Abram4}. A further cohomological framework for contextuality, which is compatible with MBQC, was described in \cite{Coho} (left leg). A cohomological formulation of MBQC (right leg) was provided in \cite{MQCcoho}. The contextual fraction \cite{ABsheaf} is a measure of the amount of contextuality present in physical settings, and it is related to the success probability of MBQCs \cite{Abram3}.

The purpose of this paper is to corroborate the relations in the left half of the above diagram, while preserving compatibility with the right half. 
 We are interested in state-dependent probabilistic contextuality proofs. Their characteristic property is that, as opposed to state-independent and state-dependent deterministic proofs, non-contextual value assignments do exist. However, no probability distribution over these value assignments reproduces the measurement statistics predicted by quantum mechanics. This is demonstrated by the violation of certain non-contextuality inequalities. 
For example, for the setting of Mermin's star it is known that a state $\rho$ is contextual w.r.t. the local observables $X_i$, $Y_i$, for $i=1,..,3$ if 
\begin{equation}\label{MeIn}
\langle X_1X_2X_3\rangle_\rho - \langle X_1Y_2Y_3\rangle_\rho - \langle Y_1X_2Y_3\rangle_\rho - \langle Y_1Y_2X_3\rangle_\rho > 2.
\end{equation} 
This is the well known Mermin inequality \cite{Merm}. It is maximized for the GHZ state, for which the above expectation value is 4. 
Here, we provide a cohomological underpinning for such probabilistic contextuality proofs. We establish the following results. 
\begin{itemize}
\item{We extend the cohomological contextuality proofs of \cite{Coho} to probabilistic scenarios.  Our results in this regard are Theorem~\ref{T_prob_beta}, and Theorem \ref{T_prob_withHamm} and  Corollary~\ref{Cor3}, invoking the cohomology of chain complexes and of groups, respectively. Our primary motivation is the relation between quantum contextuality and measurement-based quantum computation \cite{AB}-\cite{RR13}.  Quantum computation, including MBQC, is typically probabilistic, and for this reason we seek cohomological contextuality proofs that apply to probablilistic settings.}
\item{We refine Theorems~\ref{T_prob_beta} and \ref{T_prob_withHamm} by invoking the contextual fraction, see Theorems~\ref{T_frac_beta} and \ref{T_frac} (also see Theorem~3 in \cite{Abram3}). Therein, the contextual fraction arises as a resource that bounds the violation of logical non-contextuality inequalities. The cohomological aspect is retained---the maximum violation as a function of the contextual fraction is a cohomological invariant. Herein lies the unification of the resource-theoretic and the cohomological perspective.}
\item{We establish a connection between the contextual fraction and the \emph{classical} cost of evaluating Boolean functions. Namely, a Boolean function can be hard to evaluate classically only if evaluating it through MBQC requires a sizeable contextual fraction; see Theorems~\ref{T2} and \ref{T2b}.}
\end{itemize}
The remainder of this paper is organized as follows. In Section~\ref{MagStat} we review  the ``magnetostatic'' perspective on quantum contextuality through cohomology \cite{Coho}. Section~\ref{BG} covers mathematical background, such as hidden variable models, the contextual fraction, and elements of cohomology. Sections~\ref{CCP} and \ref{CCS} contain our cohomological contextuality proofs for probabilistic state-dependent probabilistic scenarios. We establish a connection between the contextual fraction and the classical cost of evaluating Boolean functions in Section~\ref{OI}. We conclude in Section~\ref{Concl}.

\section{Quantum contextuality as seen from magnetostatics}\label{MagStat}

Parity proofs of contextuality, such as Mermin's square and star \cite{Merm}, have a cohomological interpretation \cite{Coho}. When formulated in this way, these  proofs bear strong semblance to a problem in magnetism. Namely, the questions of the existence of a non-contextual value assignment and of the existence of a globally defined vector potential have essentially the same mathematical formulation. 

To illustrate this similarity, let's consider the example of Mermin's star; see Fig.~\ref{MermSq}a. Can the ten Pauli observables of the star carry consistent pre-determined measurement outcomes $\pm1$? This is not the case; an algebraic obstruction prevents it. We assume that the reader is familiar with Mermin's original argument \cite{Merm}, and do not reproduce it here.

The cohomological version of this argument is as follows. The ten observables in the star are assigned to the edges in the tessellation of the surface of a torus; See Fig.~\ref{MermSq}b. Any value assignment $s$ of an ncHVM (assuming it exists) is a function that maps a given edge $a$  to a value $s(a) \in \mathbb{Z}_2$, with the interpretation that $(-1)^{s(a)}$ is the eigenvalue obtained in the measurement of the corresponding  Pauli observable $T_a$. From the cohomological point of view, $s$ is a 1-cochain. Denote by $f$ any of the five elementary faces of the surface shown in Fig.~\ref{MermSq}b, such that $\partial f =a+b+c+d$, for four edges $a$, $b$, $c$, $d$. Then there is a binary-valued function $\beta$ defined on the faces $f$ such that $T_aT_bT_cT_d = (-1)^{\beta(f)}I$, and the operators $T_a$, $T_b$, $T_c$, $T_d$ pairwise commute. As in Mermin's original argument, these product constraints among commuting observables induce constraints among the corresponding values, namely $s(a)+s(b)+s(c) + s(d)\mod 2 =\beta(f)$. By applying this relation to the five faces of the torus, we reproduce the five constraints of Mermin's star.

\begin{figure}
\begin{center}
\begin{tabular}{lclcl}
(a) && (b) && (c) \\
\includegraphics[height=3.5cm]{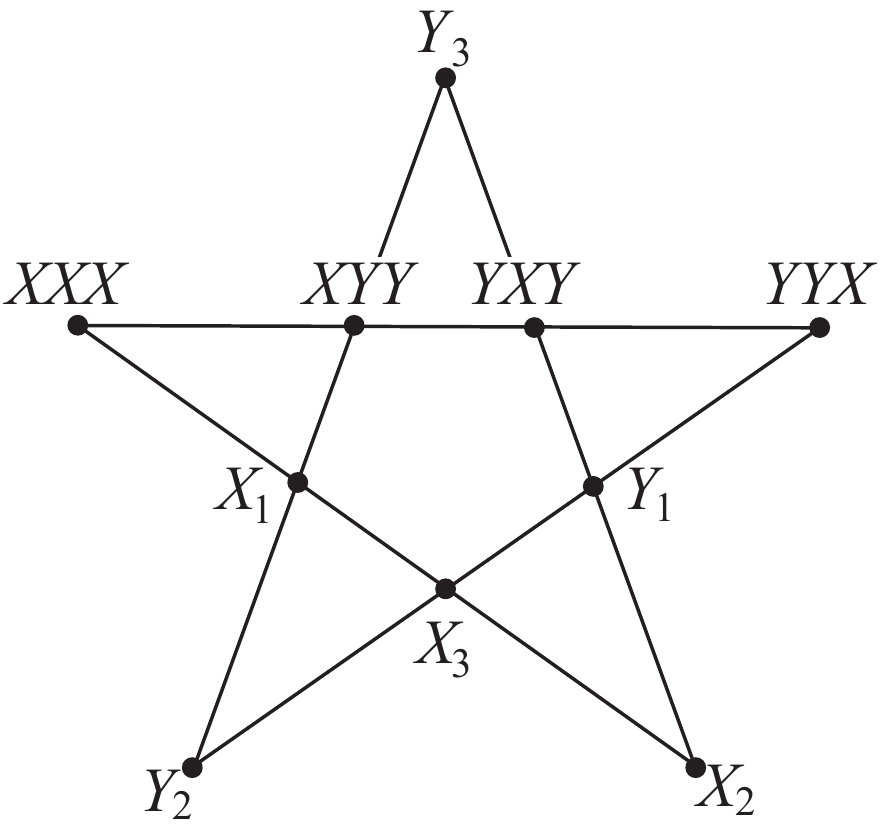} &&
\includegraphics[height=3.3cm]{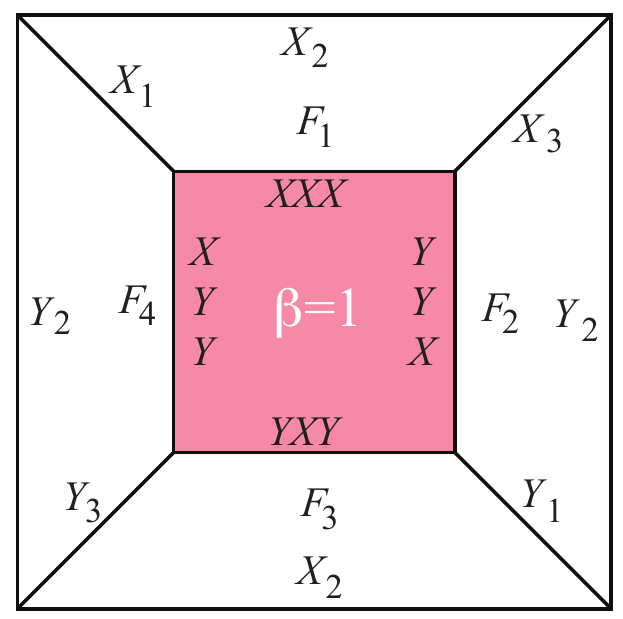} &&
\includegraphics[height=3.3cm]{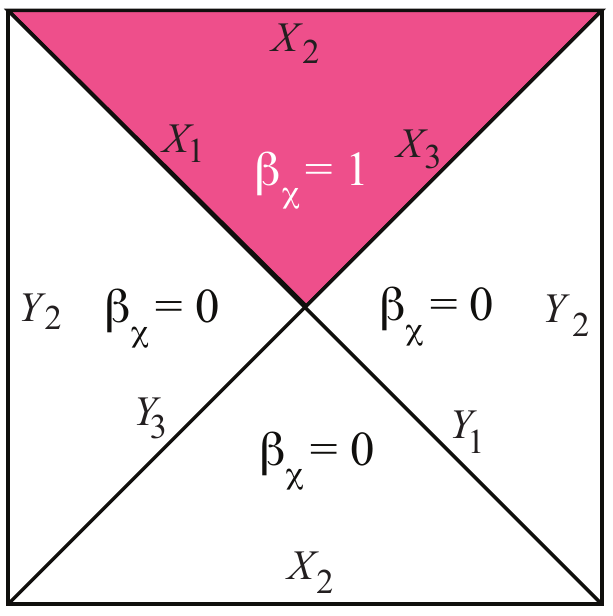}
\end{tabular}
\caption{\label{MermSq} Co-chains evaluated by a boundary operator in service of contextuality and magnetostatic proofs. (a) The Mermin star, standard representation. Each horizontal and vertical line corresponds to a measurement context, composed of four commuting Pauli observables multiplying to $\pm I$. (b) Mermin's star re-arranged on a surface. The Pauli observables are now associated with the edges, and each measurement context with the boundary of one of the five elementary faces. The exterior edges are identified as shown. 
(c) The relative complex ${\cal{C}}(E,E_0)$ for Mermin's star.}
\end{center}
\end{figure}

These constraints have a topological interpretation. Namely, $\beta$ can be interpreted as a 2-cochain. Furthermore, for any consistent context-independent value assignment $s$, the constraints between the value assignments and the function $\beta$ are given by the equation
\begin{equation}\label{beta1}
ds = \beta.
\end{equation}
Therein, $d$ the coboundary operator and the addition is $\text{mod}\; 2$. 

We can now show that for the present function $\beta$, no value assignment $s$ can satisfy Eq.~(\ref{beta1}). Namely, we observe that $\beta$ evaluates to 0 on four faces and to 1 on one face. Therefore, the integral of $\beta$ over the whole surface $F$ equals 1. Finally we note that $F$ is a 2-cycle, $\partial F = 0$. Putting all this information into Stokes' theorem (with all integration mod 2),
$$
1 = \int_F \beta  =\int_F ds = \oint_{\partial F} s =\oint_0 s =0.
$$
Contradiction. This is exactly Mermin's original argument demonstrating the non-existence of non-contextual value assignments, but in cohomological guise.

The above reasoning is not confined to Mermin's star. Rather, it applies to all parity proofs. The observables in such proofs do not need to be Pauli observables; the only requirement is that all their eigenvalues can be written in the form $\omega^z$, where $\omega:=e^{i2\pi/d}$ and $z\in \mathbb{Z}_d$, for some positive integer $d$. The general statement is the following \cite{Coho}. Every parity proof of contextuality boils down to a chain complex with a 2-cocycle $\beta$ defined on it. If the corresponding cohomology class is non-trivial, $[\beta]\neq 0$, then the setting is contextual.\smallskip

What is the connection of contextuality to magnetostatics?---The flux created by a magnetic monopole is an obstruction to the existence of a global vector potential in the same way as the above ``flux'' $\int_F \beta$ is an obstruction to the existence of a non-contextual value assignment. In more detail, consider the question of whether a given magnetic field $\textbf{B}$ can be written as the curl of some vector potential $\textbf{A}$, i.e., $\textbf{B} = \nabla \times \textbf{A}$. This possibility is ruled out by the existence of a closed surface $F$ for which $\int_F d\textbf{F}\cdot \textbf{B} \neq 0$. Here, $\textbf{A}$ is a 1-cochain (1-form) and $\textbf{B}$ is a 2-cochain (2-form). They are the counterparts of the value assignment $s$ and the function $\beta$, respectively. The magnetic flux $\int_F d\textbf{F}\cdot \textbf{B} \neq 0$ through some closed surface $F$---the counterpart of a contextuality proof $\beta(F) \neq 0$---would indicate (when observed) the presence of a magnetic monopole.\medskip 
 
To prepare for the scenarios of interest for the present work, we make, for the example of Mermin's square, the transition from state-independent to the state-dependent scenario. It is based on the same cohomological interpretation as the state-independent case; see Fig.~\ref{MermSq}b. The additional ingredient is the Greenberger-Horne-Zeilinger (GHZ) state, which is a joint eigenstate of the four non-local observables in Mermin's star, $X_1X_2X_3$, $X_1Y_2Y_3$, $Y_1X_2Y_3$ and $Y_1Y_2X_3$, with eigenvalues $1,-1,-1,-1$, respectively. We thus have the partial value assignment
\begin{equation}\label{PartAss}
s(a_{XXX})=0,\; s(a_{XYY})=1,\; s(a_{YXY})=1,\;s(a_{YYX})=1.
\end{equation}
This value assignment cannot be extended to all observables in the star, as we now show. Denote $F'=F_1+F_2+F_3+F_4$, see Fig.~\ref{MermSq}b for the labeling. Then, assuming that a value assignment exists that satisfies the relation Eq.~(\ref{beta1}), we have that 
$0 = \int_{F'} \beta = \int_{F'} ds = \int_{\partial F'} s = 1$. Contradiction. Hence, there are no non-contextual value assignments in this setting.

The present paper deals with state-dependent scenarios where the quantum state in question does not permit partial deterministic value assignments as in Eq.~(\ref{PartAss}), and where contextuality inequalities such as Eq.~(\ref{MeIn}) apply. We provide a topological underpinning for these inequalities.

\section{Mathematical background}\label{BG}

In this section, we define the notion of ``non-contextual hidden variable model'' that we will subsequently refer to, review the notion of the contextual fraction \cite{ABsheaf}, and provide necessary background on the cohomology of chain complexes and of groups.

 \subsection{Non-contextual hidden variable models}
 
We formalize the classical idea of a hidden variable model for a system, in the same manner as \cite{Coho}. Quantum states are described by density matrices $\rho$, the prescribed set of observables is $\oO$, and $M \in \mathcal{M}$ denote contexts of commuting observables in $\oO$. 

\begin{Def}\label{def_nc_hvm}
A non-contextual hidden variable model is a triple $(S,q_{\rho}, \Lambda)$, with $q_{\rho}$ a probability distribution over a set $\mathcal{S}$ of internal states. The set $\Lambda=\{\lambda_{\nu}\}_{\nu \in \mathcal{S}}$ consists of functions, $\lambda_{\nu}: \oO \rightarrow \mathbb{C}$ obeying the following constraints:

\begin{enumerate}
\item For any set $M \subset \oO$ of commuting observables there exists a quantum state $\ket{\psi}$ such that:
\begin{equation}\label{eq:qconsis}
A\ket{\psi}= \lambda_{\nu}(A) \ket{\psi}, \forall A \in M
\end{equation}

\item The distribution $q_{p}$ satisfies:

\begin{equation}\label{eq:trq}
\text{tr}(A \rho)= \sum_{\nu \in \mathcal{S}} \lambda_{\nu}(A) q_{\rho}(\nu), \forall A \in \oO
\end{equation}
\end{enumerate}
\end{Def}
From condition (\ref{eq:qconsis}) it follows that for any triple of commuting observables $A,B, AB \in \oO$, the functions $\lambda_{\nu}$ obey
\begin{equation}\label{LambdaConstr}
\lambda_{\nu}(AB)= \lambda_{\nu}(A) \lambda_{\nu}(B).
\end{equation}

 \subsection{The contextual fraction} 
 
An empirical model predicts the outcome distributions for compatible joint measurements on a physical state \cite{ABsheaf}. Such models can be used to describe quantum mechanical systems, among other things, and this is what we use them for here. An empirical model $e$ assigns an outcome probability distribution $e_M$ to every set $M$ of compatible measurements. The probability distributions $e_M$ have to satisfy consistency conditions; essentially they need to be compatible under marginalization \cite{ABsheaf}.

From the perspective of contextuality, one may ask how much of an empirical model $e$ can be described by a non-contextual hidden variable model (ncHVM). Splitting the model $e$ into a contextual part $e^C$ and a non-contextual part $e^{NC}$,
\begin{equation}
e=\lambda e^{NC} + (1-\lambda) e^C,\; 0 \leq \lambda \leq 1,
\end{equation}
we want to know what the maximum possible value of $\lambda$ is. This maximum value is called the non-contextual fraction ${\sf{NCF}}(e)$ of the model $e$,
\begin{equation}
{\sf{NCF}}(e) := \max_{e^{NC}} \lambda.
\end{equation}
The contextual fraction ${\sf{CF}}(e)$ is then defined to be the probability weight of the contextual part $e^{C}$,
\begin{equation}
{\sf{CF}}(e):=1-{\sf{NCF}}(e).
\end{equation}

\subsection{Cohomology of chain complexes} 
 
In \cite{Coho} a cohomological framework is introduced to study contextuality proofs. We first recall some notions from this framework and present a generalization which is suitable for probabilistic scenarios. Our approach is to generalize the underlying cohomological structure of state-dependant deterministic scenarios. In the deterministic case the cohomological basis of such scenarios consists of a relative complex $\cC(E,E_\Psi)$ which depends on a given state $\ket{\Psi}$. The operators corresponding to the labels in $E_\Psi$ stabilizes the resource state $\ket{\Psi}$. In the symmetry-based version there is a symmetry group acting on the labels with the extra condition on  the transformed eigenvalues. In the present framework we will start with a pair $E_0\subset E$ where $E_0$ replaces $E_\Psi$. The eigenvalues are replaced by a function $\chi$ defined on $E_0$, and a symmetry group is required to preserve $\chi$.

Let $\oO$ denote a set of observables of the form $\set{\omega^k T_a|\; a\in E,\; k \in \Z_d}$ where $E$ is a set of labels for the observables under consideration. We say $a,b\in E$ commutes whenever the corresponding operators commute $T_aT_b=T_bT_a$. The operator $T_a$ has the eigenvalues given by $\set{\omega^k|\;k\in \Z_d}$. 
For commuting observables $T_a$, $T_b$ the operators multiply as
\begin{equation}\label{prod_T}
T_{a+b} = \omega^{\beta(a,b)} T_aT_b.
\end{equation}
This gives a corresponding addition operation for the label set. Given commuting $a,b\in E$ the sum $a+b$ is defined using Eq.~(\ref{prod_T}).  

The main object in \cite{Coho} is the (co)chain complex $\cC(E)$.
For the construction of this complex $E$ is required to 
satisfy the property that $a+b\in E$ for commuting labels $a,b\in E$.
 Compared to  \cite{Coho},   we modify the definition of the chain complex so that it applies to arbitrary $E$.
The definition of $C_0(E)$ and $C_1(E)$ remains the same but we change $C_2(E)$ and $C_3(E)$.

The chain complex ${\cal{C}}_*(E)$ consists of one vertex, and edges, faces and volumes. It is constructed as follows.
\begin{enumerate}
\item $C_0(E)=\mathbb{Z}_d$, geometrically we have a single vertex.
\item $C_1(E)$ is freely generated as a $\mathbb{Z}_d$-module by the elements $[a]$ where $a\in E$. These labels correspond to the set of edges.
\item $C_2(E)$ is freely generated as a $\mathbb{Z}_d$-module by the pairs $[a|b]$ where $a,b\in E$ commutes and $a+b\in E$. The pairs $(a,b)$ correspond to faces. We denote the set of all faces by $F$.
\item $C_3(E)$ is freely generated as a $\mathbb{Z}_d$-module by the triples $[a|b|c]$ where $a,b,c\in E$ pair-wise commute and the labels $a+b$, $b+c$, and $a+b+c$ belong to $E$. These triples $(a,b,c)$ correspond to volumes and the set of volumes will be denoted by $V$.
\end{enumerate} 
The differentials in the complex
$$ 
C_3(E)\stackrel{\partial}{\rightarrow} C_2(E) \stackrel{\partial}{\rightarrow} C_1(E) \stackrel{\partial}{\rightarrow} C_0(E)
$$
are defined as before
$$
\partial[a]=0,\;\; \partial[a|b]=[b]-[a+b]+[a],\;\; \partial[a|b|c]=[b|c]-[a+b|c]+[a|b+c]-[a|b].
$$
 Given the chain complex $\cC_*(E)$, there is a corresponding cochain complex $\cC^*(E)$ as usual.
 The cochains $C^n(E)$ are $\Z_d$-linear maps $C_n(E)\rightarrow \Z_d$. Equivalently we can think of the cochains as functions on the basis elements of $C_n(E)$. For example in degree $3$ we have that $C^3(E)$ is given by the set of functions $\gamma:V\rightarrow \Z_d$. The abelian group structure on the cochains are obtained by addition of functions: Given   $\gamma,\gamma'$ the sum $\gamma+\gamma'$ is the function defined by $(\gamma+\gamma')(v)=\gamma(v)+\gamma'(v)$ for all $v\in V$. Similarly lower degree cochains can be described as functions on the basis elements, and the abelian group structure is given by addition of functions. The coboundary operator $d:C^n(E)\rightarrow C^{n+1}(E)$ is defined by $d\alpha(x)=\alpha(\partial(x))$ where $\alpha\in C^n(E)$ and $x\in C_{n+1}(E)$.

From Eq.~(\ref{prod_T}) and the above definition of ${\cal{C}}(E)$ it is clear that $\beta$ is a 2-cochain in ${\cal{C}}^*(E)$. We recall from \cite{Coho} the following properties of $\beta$.
\begin{Lemma}[\cite{Coho}]\label{BetaProp}
$\beta$ is a 2-cocycle, $d\beta =0$. Furthermore, if a   value assignment $\ss:E\longrightarrow \mathbb{Z}_d$ exists, then $\beta$ is trivial,
\begin{equation}\label{beta_ds}
\beta = - d\ss.
\end{equation}
\end{Lemma}
This lemma is the content of Eq.~(12)  and Lemma~2 in \cite{Coho}. Eq.~(\ref{beta_ds}) is just the condition Eq.~(\ref{LambdaConstr}) restated in cohomological fashion, using the definition Eq.~(\ref{prod_T}) of $\beta$. The cocycle condition $d\beta=0$ is a consequence of the associativity of operator multiplication, $(T_aT_b)T_c=T_a(T_bT_c)$.

Lemma~\ref{BetaProp} describes state-independent contextuality proofs. To the present purpose it is just an introduction. Here we are interested in the state-dependent case, and more specifically, in the probabilistic state-dependent case. The deterministic state-dependent case was already treated in \cite{Coho}, and therein, an important role is played by the set $E_\Psi \subset E$ corresponding to the stabilizer of the state $|\Psi\rangle$ in ${\cal{O}}$. In the present probabilistic scenario, the stabilizer of the state $\rho$ in ${\cal{O}}$ is generally trivial, i.e., it consists of the identity operator only. Nonetheless, the set $E_\Psi$ has a non-trivial counterpart $E_0$ in the present discussion, which we now introduce. 

$E_0 \subset E$ chosen such that two properties hold: (i) After removing from Eq.~(\ref{LambdaConstr}) all constraints that only involve observables $T_a$ with $a \in E_0$, the parity obstruction to the existence of value assignments disappears, and hence value assignments can exist. (ii) The resulting ncHVMs imply non-contextuality inequalities involving the expectation values $\langle T_a\rangle$, $a \in E_0$, which are violated by quantum mechanics. Our goal is to construct such inequalities.

For concreteness, let us look at the example of the state-dependent Mermin star. In this case, $E_0=\{a_{XXX},a_{XYY},a_{YXY},a_{YYX}\}$. The belonging constraint $\ss(a_{XXX})+ \ss(a_{XYY})+ \ss(a_{YXY})+ \ss(a_{YYX}) = 1 \mod 2$ is removed, and in result, non-contextual value assignments become possible. They imply the Mermin inequality $\langle X_1X_2X_3\rangle +  \langle X_1Y_2Y_3\rangle + \langle Y_1X_2Y_3\rangle  + \langle Y_1Y_2X_3\rangle \leq 2$ for ncHVMs. It is violated by quantum mechanics.\smallskip 

We now describe the cohomological underpinning for the probabilistic state-dependent case.  We can construct the chain complex $\cC_*(E_0)$ for the subset $E_0$. The inclusion $E_0\subset E$ gives an inclusion of the chain complexes $\cC_*(E_0)\subset \cC_*(E)$. 
Geometrically we can collapse the edges, faces, and volumes coming from $E_0$ and look at the resulting space.
In terms of chain complexes this idea is expressed using the language of relative complexes. The relative complex $\cC_*(E,E_0)$ is defined as the quotient $\cC_*(E)/\cC_*(E_0)$ meaning that in each degree $C_n(E,E_0)$ is given by the quotient group $C_n(E)/C_n(E_0)$. The basis is obtained by erasing the basis elements of $C_n(E_0)$ from the basis elements of the larger complex $C_n(E)$.
 The relative boundary operator $\partial_R$ is induced from the boundary operator $\partial$ of $\cC_*(E)$, and in effect it can be calculated by applying $\partial$ and removing the chains which lie in $\cC_*(E_0)$.   

The relation between the subcomplex and the relative complex is expressed as an exact sequence
$$
0\rightarrow \cC_*(E_0) \rightarrow \cC_*(E) \rightarrow \cC_*(E,E_0)\rightarrow 0,
$$
and similarly, there exists a corresponding exact sequence for the cochain complexes
$$
0\rightarrow \cC^*(E,E_0) \rightarrow \cC^*(E) \rightarrow \cC^*(E_0)\rightarrow 0.
$$
 The relative cochain complex $C^n(E,E_0)$ consists of cochains in $C^n(E)$ whose restriction to $C_n(E_0)$ is zero. The relative coboundary operator is the same as the coboundary operator of $\cC^*(E)$.
We studied both of these constructions in \cite{Coho} for a special subset $E_\Psi$ associated to a given state $\ket{\Psi}$.

We fix a  partial value assignment $\chi:E_0\rightarrow \Z_d$ on $E_0$, and ask whether it can be extended to all of $E$.   In practice the chain complex of $E_0$ will be one dimensional i.e. $C_n(E_0)=0$ for $n=2,3$. Although our results work for general $E_0$ we will make this assumption throughout.
Under this assumption a  partial value assignment on $E_0$ is simply a function, since there are no faces imposing compatibility. 
With respect to $\chi$ we can begin our discussion of relative complexes by defining
\begin{equation}\label{beta_rel}
\beta_\chi = \beta + d\bar\chi 
\end{equation}
where $\bar\chi$ is the extension of $\chi$  to $E$ by setting $\chi(a')=0$ for all $a'\in E-E_0$. 
We can regard $\beta_\chi$ as a cochain in the relative complex $\cC^*(E,E_0)$ since it vanishes on  $C_2(E_0)$.

\begin{Lemma}\label{chiCocyc}
The cochain $\beta_\chi$ is a cocycle, $d\beta_\chi=0$.
\end{Lemma}

\begin{proof}
We are working with relative complexes hence the coboundary is defined with respect to the relative boundary $\partial_R$. For  $v\in C_3(E)$ we have
$$
d\beta_\chi(v) = \beta_\chi(\partial_Rv)= \beta_\chi(\partial v)- \beta_\chi(\partial v-\partial_Rv)=\beta_\chi(\partial v)= \beta(\partial v) + d\bar\chi(\partial v)= 0,
$$
since $\beta$ vanishes on $\partial v-\partial_Rv$. In the last equality we used the fact that $\beta$ vanishes on boundaries (as proved in Section 4.2 of \cite{Coho}), and $d\bar\chi(\partial v)=dd\bar\chi(v)=0$.  
\end{proof}

\begin{Theorem}\label{Thm_beta}
A value assignment $\mathfrak{s}:E  \longrightarrow \mathbb{Z}_d$ with $\mathfrak{s}|_{E_0} = \chi$ exists only if $[\beta_\chi]=0$  in $H^2(E,E_0)$.
\end{Theorem}
A value assignment $\chi$ on $E_0$ cannot be extended to $E$ if $[\beta_\chi]\neq0$.  
\begin{proof}
Assume that there exists a value assignment $\ss$  for $E$ that satisfies $\ss|_{E_0}=\chi$. Now let  $s= \ss - \overline{\chi}$. Thus, $s|_{E_0} = 0$, and hence $s$ lives in the relative complex $C^1(E,E_0)$. Further, $ds = d\ss - d\overline{\chi} = - \beta - d\overline{\chi} = -\beta_\chi$, and thus $[\beta_\chi]=0$. 
\end{proof}

\subsection{Symmetry and group cohomology} 

A symmetry group $G$ is a transformation $\oO\rightarrow \oO$ that acts on $T_a$ by the equation 
\begin{equation}\label{action}
g(T_a) = \omega^{\tilde \Phi_g(a)} T_{ga}
\end{equation}
and satisfies  $g(AB) = g(A)g(B)$  for commuting operators $A,B,AB \in \oO$. 
If $\ss:E\rightarrow \Z_d$ is   a value assignment then the function defined as
\begin{equation}\label{s_update}
g\cdot \ss(a) := \ss(ga) + \tilde\Phi_g(a)
\end{equation}
is a value assignment, too \cite{Coho}.
The approach in \cite{Coho} is to interpret  $\tilde\Phi$ as a cocycle living in a suitable complex. We generalize this approach in a way that is applicable to probabilistic scenarios extending the deterministic case.

 Let $H\subset G$ be a subgroup of our symmetry group which preserves the set $E_0$ and satisfies 
\begin{equation}\label{H_chi}
 h\cdot \chi =\chi 
\end{equation}
for all $h\in H$. In the relative version we define the cochain 
\begin{equation} \label{phi_rel}
\tilde{\Phi}_\chi = \tilde{\Phi} + d^h\bar\chi
\end{equation}
where $\bar{\chi}$ is the extension of $\chi$ as before,   and $d^h$ denotes the group cohomology coboundary: $d^h\bar\chi(g,a)=\chi(ga)-\chi(a)$ for all $g\in H$ and $a\in E$.  
We will regard $\tilde{\Phi}_\chi$ as a cochain in a group cohomology complex. Next let us describe the complex.  The $H$ action on $E$ given in Eq.~(\ref{action}) induces an action on $C_q(E,E_0)$ and $C^q(E,E_0)$ where $0\leq q\leq 3$. Then we can consider the complex $C^p(H,C^q(E,E_0))$ for a fixed $q$. Here the coefficient module $M=C^q(E,E_0)$ has non-trivial action. Note that the cohomology group $H^*(H,M)$ can be defined for any $\Z_d$--module $M$ with an action of $H$ \cite{Brown}.
For a fixed $q$ the group cohomology cochain complex is given by
$$
C^q(E,E_0)=C^0(H,C^q(E,E_0)) \stackrel{d^h}{\longrightarrow} C^1(H,C^q(E,E_0)) \stackrel{d^h}{\longrightarrow} C^2(H,C^q(E,E_0)) \rightarrow \cdots
$$  
where $d^h$ will be referred to as the horizontal coboundary. Our objects of interest are as follows: $s\in C^1(E,E_0)$, $\beta_\chi\in C^2(E,E_0)$, and $\tilde{\Phi}_\chi \in C^1(H,C^1(E,E_0))$.
The group cohomology coboundary on $s$ and $\beta_\chi$ is given by
$$
d^hs(g,a) = s(ga)-s(a),\;\;\;\; d^h\beta_\chi(g,f)=\beta_\chi(gf)-\beta_\chi(f)
$$
and on $\tilde{\Phi}_\chi$ we have
$$
 d^h\tilde{\Phi}_\chi(g_1,g_2,a)= \tilde{\Phi}_\chi(g_1,g_2a)-\tilde{\Phi}_\chi(g_1g_2,a)+\tilde{\Phi}_\chi(g_2,a).
$$
Instead of fixing $q$ we can fix $p$ and construct a cochain complex using the coboundary of the relative complex $\cC(E,E_0)$ to obtain
$$
C^p(H,\Z_d)=C^p(H,C^0(E,E_0)) \stackrel{d^v}{\longrightarrow} C^p(H,C^1(E,E_0)) \stackrel{d^v}{\longrightarrow} C^p(H,C^2(E,E_0)) \rightarrow \cdots
$$  
where $d^v$ is the vertical coboundary.
For example,   we have $d^v\tilde\Phi_\chi(g,f)=\tilde\Phi_\chi(g,\partial_Rf)$ where $\partial_R$ is the relative boundary map.

\begin{Lemma}\label{phi_property}
The cochain $\tilde{\Phi}_\chi $ defined by Eq.~(\ref{phi_rel}) is a cocyle (with respect to $d^h$) in  $C^1(H,C^1(E,E_0))$ and satisfies 
\begin{equation}\label{PhiBeta}
d^v\tilde{\Phi}_\chi = d^h\beta_\chi.
\end{equation}
\end{Lemma}
\begin{proof}
 Eq.~(\ref{H_chi}) and (\ref{phi_rel}) imply that 
 \begin{equation}\label{PhiChiInvar}
 \tilde{\Phi}_\chi(g,a)=\tilde{\Phi}_g(a) + \chi(ga)-\chi(a)=0,\; \forall a\in E_0,\, \forall g\in H.
 \end{equation}
 That is the function $\tilde{\Phi}_\chi(g,-)$ vanishes on $C^1(E_0)$, hence  belongs to $C^1(E,E_0)$ by definition of the relative complex. Therefore $\tilde{\Phi}_\chi$ is a cochain in $C^1(H,C^1(E,E_0))$. 
For the cocycle property we check that the group cohomology coboundary $d^h$ vanishes:
$$
d^h\tilde{\Phi}_\chi = d^h\tilde{\Phi} + d^hd^h\bar\chi =0
$$ 
where we used $d^h\tilde{\Phi}=0$ (Lemma 3 Eq.~(31a) in \cite{Coho}) and $d^hd^h=0$. For the second property we calculate
$$
d^v\tilde{\Phi}_\chi = d^v\tilde{\Phi} + d^vd^h\bar\chi = d^h\beta+d^vd^h\bar\chi=d^h(\beta+d^v\bar\chi)=d^h\beta_\chi
$$
using  $d^v\tilde{\Phi}=d^h\beta$ (Lemma 3 Eq.~(31b) in \cite{Coho}) and $d^hd^v=d^vd^h$. 
\end{proof}

Next we reduce our symmetry group. Let $N\subset H$ denote the normal subgroup of symmetry elements which fix each element of $E$. The quotient group $Q=H/N$ is the essential part of the symmetry which acts on the complex. 
 Let $\pi:H\rightarrow Q$ denote the quotient homomorphism.
Furthermore, we need to restrict to boundaries in the relative complex.
Let $B_1\subset C_1(E,E_0)$ denote the image of $C_2(E,E_0)$ under the relative boundary operator. Let $U_0$ denote the dual of $B_1$ in the sense that it consists of $\Z_d$--linear maps $B_1\rightarrow \Z_d$. We have a surjective map $C^1(E,E_0) \rightarrow U_0$. 
 We define $\Phi_\chi$ to be the composition 
 $$
\Phi_\chi: Q\stackrel{\theta}{\longrightarrow} H \stackrel{\tilde\Phi_\chi}{\longrightarrow}  C^1(E,E_0) \rightarrow U_0
 $$
where $\theta$ is a section $Q\rightarrow H$ of the quotient map. 
 Unravelling the definition we have
$$
\Phi_\chi(q, \partial_R f) = \tilde{\Phi}_\chi(\theta(q),\partial_Rf) = d^v \tilde{\Phi}_\chi(\theta(q),f)
$$
where $q\in Q$ and $f\in C_2(E,E_0)$. The quotient map $\pi:H\rightarrow Q$ induces a map of cohomology groups
$$
\pi^*: H^1(Q, U_0) \rightarrow H^1(H,U_0)
$$
and   $[\Phi_\chi]$ maps to the class of $d^v\tilde{\Phi}_\chi$ under this map. Using Lemma \ref{chiCocyc} and \ref{phi_property} we summarize the relation between $\beta_\chi$, $\tilde{\Phi}_\chi$, and $\Phi_\chi$ as follows
%$$
%\begin{tikzcd}
%0 &  & \\
%\beta_\chi \arrow[u,mapsto,"d^v"]\arrow[r,mapsto,"d^h"] &  \pi^* \Phi_\chi & \\ 
%& \tilde{\Phi}_\chi  \arrow[u,mapsto,"d^v"] \arrow[r,mapsto,"d^h"] & 0
%\end{tikzcd}
%
%\begin{array}{cc}
%\includegraphics[width=5cm]{ArrowDiagram}\\
%\text{[replacement]}
%\end{array}
%$$ 

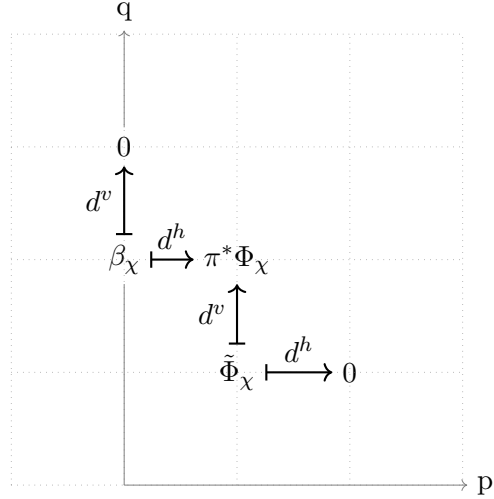
\begin{figure}
\centering
\begin{tikzpicture}[scale=1.5]
\node (beta) at (0,2) {$\beta_\chi$};
\node (dvbeta) at (0,3) {$0$};
\node (dhPhi) at (2,1) {$0$};
\node (piPhi) at (1,2) {$\pi^* \Phi_\chi$};
\node (Phi) at (1,1) {$\tilde{\Phi}_\chi$};
\node (p) at (3.2,0) {p};
\node (q) at (0,4.2) {q};
\draw [help lines,<-]   (p) -- (0,0)  ;
\draw [help lines]    (0,0)-- (beta) ;
\draw [help lines, ->] (dvbeta) edge (q);
   \draw [help lines, dotted] (-1,0) grid (3,4);
   \draw[|->,thick]
  (beta) edge node[left] {$d^v$}  (dvbeta);
  \draw[|->,thick]
  (beta) edge node[above]  {$d^h$}  (piPhi);
  \draw[|->,thick]
  (Phi) edge node[left]  {$d^v$}  (piPhi);
  \draw[|->,thick]
  (Phi) edge node[above]  {$d^h$}  (dhPhi);
  % \draw[ thick,latex-latex] (-1,4) -- (4,-6) node[anchor=south west] {$a$}; % two points for drawing 2x+y=2
  %\tkzText[above](0,6.75){Desired Output}
  \end{tikzpicture}
  \caption{The complex $C^p(H,C^q(E,E_0))$ has two types of coboundaries: horizontal $d^h$, and vertical $d^v$. The cochains $\beta_\chi$ and $\tilde\Phi_\chi$ live in degrees $(p,q)=(0,2)$ and $(1,1)$, respectively.
  }
\end{figure}

\begin{Theorem}\label{Thm_coho}
Given $(E,E_0,\chi)$ and a symmetry group $H$ satisfying $h\cdot\chi=\chi$ for all $h\in H$ if the class $[\Phi_\chi]\not=0$ in $H^1(Q,U_0)$ then $[\beta_\chi]\not=0$ in $H^2(\cC(E,E_0))$. 
\end{Theorem}
\begin{proof}
We will show that $[\beta_\chi]=0$ implies $[\Phi_\chi]=0$. Assume that $\beta_\chi = - d^vs$ for some $s\in C^1(E,E_0)$. For $q\in Q$ and $f\in C_2(E,E_0)$ we have
$$
\Phi_\chi(q,\partial_Rf)=\tilde\Phi_\chi(\theta(q),\partial_Rf)=d^v\tilde\Phi_\chi(\theta(q),f)=d^h\beta_\chi(\theta(q),f)=
-d^hd^vs(\theta(q),f)=s(\partial_Rf)-s(q\partial_Rf)
$$
where we used   Eq.~(\ref{PhiBeta}) in Lemma \ref{phi_property} and $\theta(q)\partial_Rf=q\theta_Rf$ since the normal subgroup $N$ fixes each element of $E$. Thus, $[\Phi_\chi]=0$   since $\Phi_\chi(q,\partial_Rf)=-d^hs(q,\partial_Rf)$.
\end{proof}

This result is the basis for the extension of the ideas used in \cite{Coho}. A special case is the state-dependent symmetry based contextuality proofs. There $\chi$ arises as $\ss_\Psi$ associated to the eigenvalues of the state. Note that taking $E_0=\emptyset$ specializes to the state-independent case $\cC(E,\emptyset)=\cC(E)$. In this paper we will introduce a probabilistic version which generalizes the deterministic scenario of state-dependent contextuality.

\section{Cohomological proofs of contextuality based on parity}\label{CCP}

We now have the tools at hand to construct cohomological proofs of contextuality for probabilistic scenarios. In this section, we provide proofs of this kind that are based on parity arguments, such as Mermin's inequality (\ref{MeIn}).
 
 We begin with the contextuality witnesses. For a subset $E_0\subset E$ and a function $\chi:E_0\rightarrow \Z_d$  we define the operator
\begin{equation}\label{P}
P_{\chi} = \frac{1}{|E_0|} \sum_{a\in E_0} P_{a,\chi(a)}
\end{equation}
where $P_{a,\chi(a)}$  denotes the projector onto the eigenspace of $T_a$ associated to the eigenvalue $\omega^{\chi(a)}$.  Explicitly, the projector has the form
$$
P_{a,\chi(a)} = \frac{1}{d} \sum_{k\in \Z_d} \omega^{-k\chi(a)} T^k_{a}.
$$
We define a  probability function 
\begin{equation}\label{DEFp}
p_{\chi}(\rho) =     \Span{P_{\chi}}_\rho
\end{equation}   
as the expectation value of $P_{\chi}$ with respect to the state $\rho$. Note that $p_\chi$ is a probability. By Eq.~(\ref{DEFp}), $0\leq p_\chi(\rho) \leq 1$, for all density operators $\rho$.

Depending on the function $\chi$, ncHVMs impose non-trivial bounds on the probabilities $p_{\chi}(\rho)$. To state these bounds and describe their cohomological properties, it is useful to introduce the notion of ``$\beta$-compatible cochains''.

\begin{Def}\label{def_beta_comp}A $\beta$-compatible cochain is a 1-cochain $\ss \in C^1(E)$ that satisfies Eq.~(\ref{beta_ds}).
\end{Def}
Thus, every ncHVM value assignment is a $\beta$-compatible cochain. The reverse is not necessarily true. While every ncHVM value assignment has to respect the constraint Eq.~(\ref{beta_ds}), it is conceivable that there are independent additional constraints on those assignments.

We denote the set of $\beta$-compatible cochains by $\overline{\Lambda}$,
\begin{equation}
\overline{\Lambda}:=\{\mathfrak{s} \in C^1(E)|\, d\mathfrak{s} = -\beta\}.
\end{equation}
With Definition~\ref{def_nc_hvm} and the above observation, we have the relation
\begin{equation}\label{LLbar}
\Lambda \subseteq \overline{\Lambda}.
\end{equation}
Another ingredient in the bounds stated below is the Hamming distance, which measures the degree of similarity between two functions. Given two functions $f,g:E_0\rightarrow \Z_d$, the Hamming distance is defined as
$$
\H(f,g) = |E_0| - \sum_{a\in E_0} \delta_{f(a),g(a)}.
$$ 
Further, let $\H(\chi,\overline{\Lambda})$ denote the minimum of $\H(\chi,\mathfrak{s}|_{E_0})$ as $\mathfrak{s}\in\overline{\Lambda}$ is varied,
$$
\H(\chi,\overline{\Lambda}):= \min_{\mathfrak{s}\in {\overline{\Lambda}}} \H(\chi,\mathfrak{s}|_{E_0}).
$$ 
Given a function $\chi:E_0\rightarrow \Z_d$, which is automatically a $\beta$-compatible cochain since $\cC_*(E_0)$ is one dimensional, we can define $\beta_\chi$ as in Eq.~(\ref{beta_rel}). It is a cocycle in the relative complex $C^2(E,E_0)$.

\begin{Theorem}\label{T_prob_beta} 
A scenario $(\oO,\rho)$ is contextual  if 
\begin{equation}\label{ncie}
p_{\chi}(\rho) > 1 -\frac{\H(\chi, \overline{\Lambda})}{|E_0|}.
\end{equation}
\end{Theorem}
\begin{proof}
Assume as given a ncHVM with value assignments $\Lambda$ and a probability distribution $q$. The ncHVM expression $p_{\chi}(q)$ for the quantity $p_{\chi}(\rho)$ satisfies
$$
\begin{array}{rcl}
p_{\chi}(q) &=& \displaystyle{\frac{1}{|E_0|} \sum_{\ss \in \Lambda,a \in E_0} q(\ss) \delta_{\chi(a),\ss(a)}} \\
&\leq& \displaystyle{\frac{1}{|E_0|} \max_{\ss\in \Lambda} \sum_{a\in E_0}  \delta_{\chi(a),\ss(a)}} \\
&\leq& \displaystyle{ \frac{1}{|E_0|} \max_{\ss\in \overline{\Lambda}} \sum_{a\in E_0} \delta_{\chi(a),\ss(a)}} \\
&=& \displaystyle{\frac{1}{|E_0|} (|E_0|- \H(\chi,\overline{\Lambda})).}
\end{array}
$$  
Therefore, if $p_{\chi}(\rho)$ is larger than $1 - \H(\chi, \overline{\Lambda})/|E_0|$ then no ncHVM can describe the given scenario $(\oO,\rho)$.
\end{proof}

Theorem~\ref{T_prob_beta} has the following implication.
\begin{Cor}\label{C_prob_beta}
A scenario $(\oO,\rho)$ is contextual  if $[\beta_\chi]\neq 0$ and
$$ 
p_{\chi}(\rho) > 1 -\frac{1}{|E_0|}.
$$
\end{Cor}
\begin{proof}
If $[\beta_\chi] \neq 0$, then Theorem~\ref{Thm_beta} says that no value assignment $\mathfrak{s}: E \longrightarrow \mathbb{Z}_d$ exists such that $\mathfrak{s}|_{E_0} = \chi$. Therefore, $\H(\chi, \overline{\Lambda}) \geq 1$. Now combining this with Theorem~\ref{T_prob_beta}, the scenario $(\oO,\rho)$ is contextual if $p_{\chi}(\rho) > 1 - 1/|E_0|$.
\end{proof}
Corollary~\ref{C_prob_beta} generally produces weaker contextuality thresholds than Theorem~\ref{T_prob_beta}. We  state it nonetheless, for two reasons: (i) It is the direct probabilistic generalization of Theorem~2 in \cite{Coho}. (ii) Through the condition $[\beta_\chi]\neq 0$ it is evident that also in probabilistic settings contextuality has a topological aspect.

The latter is not a priori clear for Theorem~\ref{T_prob_beta}, and Corollary~\ref{C_prob_beta} thus prompts the question ``Is the Hamming distance $\H(\chi, \overline{\Lambda})$ a cohomological invariant?''---This turns out to be the case.
\begin{Theorem}\label{Hamming}
The Hamming distance $\H(\chi, \overline{\Lambda})$ is a cohomological invariant, $\H(\chi,\overline{\Lambda})=\H(\chi',\overline{\Lambda})$ if $[\beta_\chi]=[\beta_{\chi'}]$.
\end{Theorem}
\begin{proof}
Assume that $\chi'$ is another value assignment on $E_0$ such that $\beta_\chi$ and $\beta_{\chi'}$ are in the same cohomology class i.e. $\beta_{\chi'}=\beta_{\chi}+ds$ for some $s\in C^1(E,E_0)$. Note that since $s$ lives in the relative complex it vanishes on $E_0$. Using the definition for $\beta_{\chi}$ and $\beta_{\chi'}$ we obtain  
\begin{equation}\label{dC}
d(s+\overline{\chi}-\overline{\chi}')=0.
\end{equation} 
Now assume a $\beta$-compatible cochain $\ss\in \overline{\Lambda}$, i.e., it holds that $d\ss=-\beta$. Now subtracting Eq.~(\ref{dC}) from the last relation, we find that
$d(\ss- s-\overline{\chi}+\overline{\chi}') = -\beta$. Hence, $\ss- s-\overline{\chi}+\overline{\chi}'$ also is a $\beta$-consistent cochain. 
By Definition~\ref{def_beta_comp} we have
\begin{equation}\label{eqLambda}
\{ \ss-s-\overline{\chi}+\overline{\chi}',\, \ss \in \overline{\Lambda}\} =\overline{\Lambda}.
\end{equation}  Then we can write
$$
\begin{array}{rcl}
\H(\chi,\overline{\Lambda}) &=& \min_{\ss \in \overline{\Lambda}} \H(\chi,\ss|_{E_0})\\
&=& \min_{\ss \in \overline{\Lambda}} \H(\chi,(\ss -s)|_{E_0})\\
&=& \min_{\ss \in \overline{\Lambda}} \H(0,(\ss-s-\overline{\chi})|_{E_0})\\
&=& \min_{\ss \in \overline{\Lambda}} \H(\chi', (\ss-s-\overline{\chi}+\overline{\chi}')|_{E_0}) \\
&=& \min_{\ss \in \overline{\Lambda}} \H(\chi', \ss|_{E_0}) \\
&= &\H(\chi',\overline{\Lambda}).
\end{array}
$$
Therein, in the first step we used the fact that $s$ vanishes on $E_0$, and in the last step we used Eq.~(\ref{eqLambda}). We have shown that $\H(\chi, \overline{\Lambda})=\H(\chi', \overline{\Lambda})$ whenever $[\beta_\chi]=[\beta_{\chi'}]$ in $H^2(\cC(E,E_0))$.
\end{proof}

{\em{Example.}} 
 We return to Mermin's star, where we have
	$$
	E- E_0  = \{a_{X_i}, a_{Y_i}, \,i=1,..,3\},\; E_0 = \{a_{X_1Y_2Y_3}, a_{Y_1X_2Y_3}, a_{Y_1Y_2X_3}, a_{X_1X_2X_3}\}
	$$
and $\eta(E_0)\subset \eta(E)$ denote the corresponding set of observables.	
	We note that the GHZ state $|\text{GHZ}\rangle= (|000\rangle + |111\rangle)/\sqrt{2}$ is an eigenstate of all observables in $\eta(E_0)$, with eigenvalues $-1,-1,-1,1$, respectively. Correspondingly, we choose the function $\chi$ that appears in the definition of $\beta_\chi$ to be
	$$
	\chi(a_{XYY}) = \chi(a_{YXY}) = \chi(a_{YYX}) = 1,\;\;  \chi(a_{XXX}) =0. 
	$$
We now show that for this function $\chi$, both Theorem~\ref{T_prob_beta} and Corollary~\ref{C_prob_beta} reproduce the Mermin inequality (\ref{MeIn}) when applied to Mermin's star. First, regarding Theorem~\ref{T_prob_beta}, one of the closest functions to $\chi$ that is induced by a $\beta$-compatible cochain $\ss$ is $\ss|_{E_0}\equiv 1$, which comes from $\ss(a_{X_3})=\ss(a_{Y_3})=1$, $\ss(a_{X_1})=\ss(a_{Y_1})=\ss(a_{X_2})=\ss(a_{Y_2})=0$. Hence, $\H(\chi,\overline{\Lambda})=1$. Thus, Theorem~\ref{T_prob_beta} says that probabilistic state-dependent version of Mermin's star is contextual for all states $\rho$ with
	\begin{equation}\label{MI}
	\frac{1}{2} + \frac{\langle X_1X_2X_3\rangle_\rho -\langle X_1Y_2Y_3\rangle_\rho - \langle Y_1X_2Y_3\rangle_\rho - \langle Y_1Y_2X_3\rangle_\rho }{8} > \frac{3}{4}.
	\end{equation}
	This reproduces the familiar Mermin inequality \cite{Merm}; cf. Inequality (\ref{MeIn}). The GHZ state violates the non-con\-textuality inequality~(\ref{MI}) maximally. 
	
Regarding Corollary~\ref{C_prob_beta}, the relative complex ${\cal{C}}(E,E_0)$ and $\beta_\chi$ for this scenario is shown in Fig.~\ref{MermSq}c. For the surface $F'$ in the figure it holds that $\partial_R F'=0$ and  $\int_{F'}\beta_\chi = 1$; hence $[\beta_\chi]\neq 0$, and Corollary~\ref{C_prob_beta} can be applied. It produces the same inequality (\ref{MI}) as Theorem~\ref{T_prob_beta}.
\medskip

Returning to the general case, we observe that by using the notion of contextual fraction we can state Theorem~\ref{T_prob_beta} in a more general form.   With our quantum setting $(\rho,{\cal{O}})$ the emprical model $e$ comes from the state $\rho$.
The contextual fraction amounts to the decomposition of $e$ into a contextual portion $e^{C}$ and a non-contextual portion $e^{NC}$,
\begin{equation}\label{RhoSplit}
e = \CFs(\rho)\, e^{C} + \NCFs(\rho)\, e^{NC}.
\end{equation}

\begin{Theorem}\label{T_frac_beta}
Consider a scenario $(\rho,{\cal{O}})$ and a restricted value assignment $\chi: E_0 \longrightarrow \mathbb{Z}_d$. Then, the probability function $p=p_{\chi}(\rho)$ satisfies 
\begin{equation}\label{psBound01}
p \leq  1-  \frac{{\sf{NCF}(\rho)}\,\H(\chi,\overline{\Lambda})}{|E_0|}.
\end{equation}
\end{Theorem}

\begin{proof}
Since quantum mechanical expectation values are linear in the state $\rho$, with Eq.~(\ref{RhoSplit}) we have 
$$  
p_{\chi}(\rho) = \CFs(\rho)\, p^{C} + \NCFs(\rho)\, p^{NC}.   
$$
Now using therein the trivial upper bound $p^{C}\leq 1$ for the contextual part, and the bound $p^{NC}\leq 1-\H(\chi,\overline{\Lambda})/|E_0|$ of Theorem~\ref{T_prob_beta} for the non-contextual part, we obtain Eq.~(\ref{psBound01}).
\end{proof}
Theorem~\ref{T_frac_beta} shows that the probability $p$ can get close to the maximal value of 1 only if the contextual fraction $\sf{CF}(\rho)$ is close to unity. More generally, the larger the contextual fraction, the larger the reachable value for $p$. To make this more explicit, we define the amount $\Delta_\chi$ of violation of the non-contextuality inequality (\ref{ncie})   as  
$$\Delta_\chi(\rho):=p_{\chi}(\rho)-  \left( 1-  \frac{\H(\chi,\overline{\Lambda})}{|E_0|} \right).$$ With Theorem~\ref{T_frac_beta} we find that
\begin{equation}\label{Mavio}
\Delta_\chi(\rho) \leq \frac{{\sf{CF}(\rho)}\,\H(\chi,\overline{\Lambda})}{|E_0|}.
\end{equation}
The amount $\Delta_\chi$ of violation of a non-contextuality inequality based on $\chi$ can only be large if the contextual fraction is large and the Hamming distance of $\chi$ to the closest function in $\overline{\Lambda}$ is large.   

The cohomological aspect of Eq.~(\ref{Mavio}) is that the map $\sf{CF}(\rho) \mapsto \text{max}\, \Delta_\chi (\rho)$ is a cohomological invariant, since $\H(\chi,\overline{\Lambda})/|E_0|$ is one by Theorem~\ref{Hamming}. In this way, Theorems~\ref{Hamming} and \ref{T_frac_beta} represent a unification of the resource-theoretic and the cohomological aspects of contextuality.

\section{Cohomological proofs of contextuality based on symmetry}\label{CCS}

In the previous section we provided cohomological contextuality proofs based on parity. The central result therein, Theorem~\ref{T_frac_beta}, is by itself not topological, but a cohomological interpretation for it is provided by Theorem~\ref{Hamming}.   In this section we will consider symmetry-based versions of these results. The Hamming distance needs to be modified in order to include the symmetry group. 
We present two results of this kind, in Sections~\ref{CP1} and \ref{CP2}. In addition, one result from Section~\ref{CCP}, Corollary~\ref{C_prob_beta}, has a direct symmetry-based counterpart, and we present it in Section~\ref{CC1}.

\subsection{Symmetry-based counterpart to Corollary~\ref{C_prob_beta}}\label{CC1}

 Recall that  we have an additional requirement for the symmetry group $H$, namely $h\cdot \chi=\chi$ for all $h\in H$ that is 
\begin{equation*}
 h(T_a) =\omega^{\chi(a)-\chi(ha)}T_{ha}\;\;\text{ for all } a\in E_0.
\end{equation*}
Then $\tilde{\Phi}_\chi$ is a cocycle in $C^1(H,C^1(E,E_0))$ by Lemma \ref{phi_property}, and Theorem \ref{Thm_coho} applies.

\begin{Cor}\label{T_prob}
Consider  a physical setting $({\cal{O}},\rho)$, with a restricted value assignment $\chi:E_0 \longrightarrow \mathbb{Z}_d$ and a symmetry group $H$  with corresponding phase function $\Phi_\chi: Q \longrightarrow U_0$ such that $[\Phi]\neq 0$ in $H^1(Q,U_0)$. This setting is contextual if it holds that
$$ 
p_{\chi}(\rho) > 1 -\frac{1}{|E_0|}.
$$
\end{Cor}
\begin{proof}
 Since $[\Phi]\neq 0$ Theorem  \ref{Thm_coho} implies that $[\beta_\chi]\not= 0$. Then we can apply Corollary \ref{C_prob_beta} to conclude that the given system is contextual.
\end{proof}

\subsection{First symmetry-based counterpart to Theorems~\ref{T_prob_beta}-\ref{T_frac_beta}}\label{CP1}

As in the parity case the bound can be improved using a suitable Hamming distance with the cost of modifying the probability function. The symmetry group $Q$ enters into the picture for both the Hamming distance and the probability function. We define the set 
\begin{equation}\label{LQdef}
\bar{\Lambda}_Q = \set{ \ss\in C^1(E)  |\; d^vd^h\ss = -d^h\beta }
\end{equation}
which will replace the role of $\bar\Lambda$.

For the symmetry-based proofs we consider $d^h\chi$ and $d^h\ss|_{E_0}$ as functions of the form $Q\times E_0 \rightarrow \Z_d$, and their Hamming distance $\H(d^h\chi, d^h\ss|_{E_0})$. We denote by $\H(d^h\chi,d^h\bar{\Lambda}_Q)$ the minimum distance as $d^h\ss$ varies in the set $d^h\bar\Lambda_Q=\set{d^h\ss|\;\ss\in \bar\Lambda_Q}$.

We now include $\H(d^h\chi,d^h\bar{\Lambda}_Q)$ in a contextuality bound. This new bound requires that the quotient group $Q$ and the set $E_0$ are such that $[qa,a]=0$, $\forall q\in Q$ and all $a\in E_0$. We define a new probability function which invokes the quotient group $Q$,
$$
p_{d^h\chi}(\rho) = \frac{1}{|Q||E_0|} \sum_{(q,a)\in Q\times E_0} \Span{P_{qa-a,\, d^h\chi(q,a)-\beta(qa,a) } }_\rho.
$$
Using $T_{qa}T^{-1}_a = \omega^{\beta(qa,a)}T_{qa-a}$ the projector can be expressed as  
$$
 P_{qa-a,\, d^h\chi(q,a)-\beta(qa,a) } =  \frac{1}{d} \sum_{k\in \Z_d} \omega^{-k(\chi(qa)-\chi(a))} (T_{qa}T^{-1}_a)^k.
$$

\begin{Theorem}\label{T_prob_withHamm}
Consider  a physical setting $({\cal{O}},\rho)$, with a restricted value assignment $\chi:E_0 \longrightarrow \mathbb{Z}_d$ and a symmetry group $H$ such that $qa$ and $a$ commute for all $q\in Q$ and $a\in E_0$. This setting is contextual if it holds that
$$ 
p_{d^h\chi}(\rho) > 1 -\frac{\H(d^h\chi,d^h\bar\Lambda_Q)}{|Q||E_0|}.
$$
\end{Theorem}

\begin{proof}
Assume that a ncHVM is provided with value assignments $\Lambda$ and a probability distribution $q$. The ncHVM expression $p_{d^h\chi}(q)$ for the quantity $p_{d^h\chi}(\rho)$ satisfies
$$
\begin{array}{rcl}
p_{d^h\chi}(q) &=& \displaystyle{\frac{1}{|Q||E_0|} \sum_{\ss \in \Lambda,a \in E_0} q(\ss) \delta_{d^h\chi(q,a)-\beta(qa,a),\ss(qa-a)}} \\
&\leq& \displaystyle{\frac{1}{|Q||E_0|} \max_{\ss\in \Lambda} \sum_{a\in E_0}  \delta_{d^h\chi(q,a),d^h\ss(q,a)}} \\
&\leq& \displaystyle{ \frac{1}{|Q||E_0|} \max_{\ss\in \overline{\Lambda}_Q} \sum_{a\in E_0} \delta_{d^h\chi(q,a),d^h\ss(q,a)}} \\
&=& \displaystyle{\frac{1}{|Q||E_0|} (|Q||E_0|- \H(d^h\chi,d^h\overline{\Lambda}_Q )).}
\end{array}
$$  
where in the second line we use $\ss(qa-a) = \ss(qa)-\ss(a)- \beta(qa,a)$ since by assumption $qa$ commutes with $a$.
Therefore, if $p_{d^h\chi}(\rho)$ is larger than $1 - \H(d^h\chi, d^h\overline{\Lambda}_Q)/|E_0|$ then no ncHVM can describe the given scenario.
\end{proof}

{\em{Example.}} 
Continuing with the Mermin star example we consider $d^h\bar\Lambda_Q$ that is the set consisting of $d^h\ss$ where $\ss\in \bar\Lambda_Q$.  
Functions in $\bar{\Lambda}_Q$ satisfy 
$\ss(a_{XXX})+\ss(a_{YYX})+\ss(a_{XYY})+\ss(a_{YXY})=0$ (similar to $\bar\Lambda$).
Then we see that the restriction of  $d^h\ss(q,-)$ to $E_0$ either maps all edges in $E_0$ to $0$ or it maps them to $1$. Taking $\chi$ as  before, $\chi(a_{XXX})=0$ and on other edges it takes the value $1$,   the Hamming distance
$$
\H(d^h\chi,d^h\bar{\Lambda}_Q)=2
$$ 
since $d^h\chi$ sends $a_{XXX},a_{YYX}$ to $1$, and $a_{XYY},a_{YXY}$ to $0$. We get the same result if we use $d^h\bar\Lambda$ instead. 
Therefore the bound gives
  $$
p_{d^h\chi}(q) \leq  1- \frac{\H(d^h\chi, d^h\bar\Lambda) }{|Q||E_0|} = 1- \frac{2}{2\cdot 4} =\frac{3}{4}
$$
as in the parity case.

We show that this Hamming distance is an invariant in group cohomology.

\begin{Theorem}\label{Hamming2}
Let $H$ and $H'$ be symmetries of the system $(E,E_0,\chi)$ and $(E,E_0,\chi')$, and $N\subset H$ and $N' \subset H'$ normal subgroups that fix the edges in $E_0$ such that $H/N=H'/N'=Q$. It holds that
if $[\Phi_\chi] = [\Phi_{\chi'}]$ then $\H(d^h\chi,d^h\bar\Lambda_Q)=\H(d^h\chi',d^h\bar\Lambda_Q)$.
\end{Theorem}
\begin{proof}
The equation $[\Phi_\chi]=[\Phi_{\chi'}]$ means that $\Phi_\chi-\Phi_{\chi'} = d^hs$ where $s\in C^1(E,E_0)$. Unravelling the definitions of $\Phi_\chi$ and $\Phi_{\chi'}$ we have
\begin{equation}\label{PhiUnravel} 
\tilde\Phi(\theta'(q),\partial f)+d^h\bar\chi'(q,\partial f) - \tilde\Phi(\theta(q),\partial f)-d^h\bar\chi(q,\partial f) = d^hs(q,\partial f) 
\end{equation}
where $\theta$ and $\theta'$ are the sections corresponding to the symmetry groups $H$ and $H'$. After pulling the relative boundary out as $d^v$ we use the relation $d^v\tilde\Phi=d^h\beta$, which allows us to forget about the sections and retain only the symmetry element $q$ in the arguments. Cancelling $d^h\beta$ we find that
$$
d^vd^h(\bar{\chi}'-\bar{\chi}-s)=0.
$$
Therefore,  given $\chi$, $\chi'$ satisfying $[\Phi_\chi] =[\Phi_{\chi'}]$, there exists an $s \in C(E,E_0)$ such that
\begin{equation}\label{LQprop}
\{\mathfrak{s}-s-\bar{\chi}+\bar{\chi}',\, \mathfrak{s} \in \overline{\Lambda}_Q\} = \overline{\Lambda}_Q.
\end{equation}
We now turn to the Hamming distance. We have
$$
\begin{array}{rcl}
\H(d^h\chi,d^h\overline{\Lambda}_Q) &=& \min_{\ss \in \overline{\Lambda}_Q} \H(d^h\chi,d^h\ss|_{E_0})\\
&=& \min_{\ss \in \overline{\Lambda}_Q} \H(0,d^h(\ss-s-\overline{\chi})|_{E_0})\\
&=& \min_{\ss \in \overline{\Lambda}_Q} \H(d^h\chi', d^h(\ss-s-\overline{\chi}+\overline{\chi}')|_{E_0}) \\
&= &\H(d^h\chi', d^h \overline{\Lambda}_Q).
\end{array}
$$
Therein, in the second line, $s|_{E_0}=0$ since $s\in C(E,E_0)$. The last line follows with Eq.~(\ref{LQprop}).
\end{proof}
We can generalize Theorem~\ref{T_prob_withHamm} by invoking the contextual fraction, in the same way as we promoted  Theorem~\ref{T_prob_beta} to Theorem~\ref{T_frac_beta}.

\begin{Theorem}\label{T_frac}  
Consider  a physical setting $({\cal{O}},\rho)$, with a restricted value assignment $\chi:E_0 \longrightarrow \mathbb{Z}_d$ and a symmetry group $H$   such that $qa$ and $a$ commute for all $q\in Q$ and $a\in E_0$. 
The probability function $p=p_{d^h\chi}(\rho)$ then satisfies 
\begin{equation}\label{psBound2} 
p \leq  1-  \frac{{\sf{NCF}(\rho)}\,\H(d^h\chi,d^h\overline{\Lambda}_Q)}{|Q||E_0|}.  
\end{equation}
\end{Theorem}
The proof of Theorem~\ref{T_frac} given Theorem~\ref{T_prob_withHamm} is the same as the proof for Theorem~\ref{T_frac_beta} given Theorem~\ref{T_prob_beta}.

\subsection{Second symmetry-based counterpart to Theorems~\ref{T_prob_beta}-\ref{T_frac_beta}}\label{CP2}

We have the following result.
\begin{Cor}\label{Cor3}
A scenario $(\oO,\rho)$ is contextual  if 
\begin{equation}\label{ncie2}
p_{\chi}(\rho) > \displaystyle{1 -\frac{\H(\chi, \overline{\Lambda}_Q)}{|E_0|}}.
\end{equation}
\end{Cor}
\begin{proof}
We recall the definitions of $\overline{\Lambda}$ and $\overline{\Lambda}_Q$, cf. Def.~\ref{def_beta_comp} and Eq.~(\ref{LQdef}). Since $d^v\mathfrak{s} = -\beta$ implies $d^vd^h\mathfrak{s} = -d^h\beta$, it holds that
$
\overline{\Lambda} \subseteq \overline{\Lambda}_Q
$. Thus, $\H(\chi, \overline{\Lambda}_Q) \leq \H(\chi, \overline{\Lambda})$, and Eq.~(\ref{ncie2}) follows with Theorem~\ref{T_prob_beta}.
\end{proof}
Again our goal is to show that the quantity on the r.h.s. of Eq.~(\ref{ncie2}) is an invariant under group cohomology.
\begin{Theorem}\label{Hamming3}
Let $H$ and $H'$ be symmetries of the system $(E,E_0,\chi)$ and $(E,E_0,\chi')$, and $N\subset H$ and $N' \subset H'$ normal subgroups that fix the edges in $E_0$ such that $H/N=H'/N'=Q$. Then,
$[\Phi_\chi] = [\Phi_{\chi'}]$ implies 
$$\H(\chi,\overline{\Lambda}_Q)=\H(\chi',\overline{\Lambda}_Q).$$
\end{Theorem}
\begin{proof}
We have
$$
\begin{array}{rcll}
\displaystyle{\H(\chi,\overline{\Lambda}_Q)} &=& \displaystyle{\min_{\mathfrak{s}\in \overline{\Lambda}_Q} \H(\chi,\mathfrak{s}|_{E_0})}\\
&=& \displaystyle{\min_{\mathfrak{s}\in \overline{\Lambda}_Q} \H(0,(\mathfrak{s}-\overline{\chi})|_{E_0})}\\
&=& \displaystyle{\min_{\mathfrak{s}\in \overline{\Lambda}_Q} \H(\chi',(\mathfrak{s}-s-\overline{\chi}+\overline{\chi}')|_{E_0})}, & \text{for some}\, s\in C(E,E_0)\\
&=& \displaystyle{\min_{\mathfrak{s}\in \overline{\Lambda}_Q} \H(\chi',\mathfrak{s}|_{E_0})}\\
&=& \displaystyle{\H(\chi',\overline{\Lambda}_Q)}\\
\end{array}
$$
Therein, in the third line we choose the particular $s\in C(E,E_0)$ that satisfies the relation $\Phi_\chi - \Phi_{\chi'} = d^hs$,  granted from the condition $[\Phi_\chi] = [\Phi_{\chi'}]$. In the fourth line we have used Eq.~(\ref{LQprop}). 
\end{proof}

\section{A computational interpretation of the contextual fraction}\label{OI}

Contextuality is for measurement-based quantum computation. This was first revealed in \cite{AB}, where the state-dependent version of Mermin's star \cite{Merm} was repurposed as a small MBQC evaluating an OR-gate. In MBQC, the evaluation of an OR gate, and, in fact, any non-linear Boolean function, requires contextuality. 

This  result can be puzzling.  Per se, there is nothing quantum about OR gates; it can hardly get any more classical in computation. If so, then how can these gates be contextual?---The resolution is that OR-gates are classical when executed by classical means, as they normally are. They require quantumness, however, when executed as MBQCs. The statement \cite{AB} does not lead to a contradiction because its domain of applicability is so narrow. Ways of evaluating Boolean functions other than MBQC, in particular classical ways, are not constrained by it. 

Yet,  there {\em{is}} a connection between the efficiencies of evaluating non-linear Boolean functions by MBQC and by purely classical means. As we show in this section, the classical memory cost of storing a Boolean function can be high only if evaluating this function through MBQC is substantially contextual. Further, in Appendix~\ref{CC} we show that, with some additional assumptions on the set $E_0$, the same holds for the operational cost of evaluating a Boolean function.
\medskip

Up to now, the function $\chi$ has merely been a label for contextuality witnesses. For some such functions the maximum violation $\Delta_\chi$ of the corresponding non-contextuality inequality is high, for other functions $\chi$ it is low, and for yet others  there is no violation at all; see Eq.~(\ref{Mavio}). There are limiting cases, such as the maximal violation of Mermin's inequality in the GHZ scenario, where the witness $p_{\chi}$ assumes its optimal value of 1. These limiting cases amount to determining the function $\chi$ by measurement of the observables $\{T_a|\, a\in E_0\}$.

Now, even away from these limiting cases, we may regard the measurement of a contextuality witness as the probabilistic evaluation of the corresponding function $\chi$ on all inputs, with average success probability $p_{\chi}(\rho)$. This observation induces a shift in how $\chi$ may be viewed, from parameter in contextuality witnesses to function computable by physical measurement. Measurement-based quantum computation pertains to the latter view, for sets $E_0$ with a special structure \cite{MQCcoho}. 

With this in mind, we consider the task of evaluating the function $\chi:E_0 \longrightarrow \mathbb{Z}_d$, by measurement of the quantum state $\rho$. To evaluate $\chi(a)\in \mathbb{Z}_d$ for any given $a\in E_0$, the observable $T_a=\eta(a)$ is measured and the corresponding outcome is reported. This is in general a probabilistic process. We may compare it to a classical process computing the function $\chi$ with the same average success probability, and ask how much information the classical process needs to have about $\chi$.

 Since the present settings allow for non-contextual value assignments, with Lemma~\ref{BetaProp} we have $[\beta]=0$. Therefore, we can choose the function $\eta$ such that $\beta\equiv 0$. We call this specific choice of function $\eta_0$.

\begin{Theorem}\label{T2}
Consider the probabilistic computation of a function $\chi: E_0 \longrightarrow \mathbb{Z}_d$, (a) by quantum means via the measurement of the observables  $\eta_0(E_0)$, and (b) by classical means. Then, the amount $I$ of information, in bits, required by the optimal classical routine (b) to compute $\chi$ with the same average success probability as the quantum routine (a) is bounded by
\begin{equation}\label{MemCo}
I  \leq  C\,  \lceil {\sf{CF}}(\rho)\, \H(\chi,\overline{\Lambda})\rceil + D,
\end{equation}
with $C=(\lceil \log_2|E_0|\rceil  + \lceil \log_2d-1\rceil)$ and $D=\lceil \log_2d\rceil \log_d\left|\overline{\Lambda}\right|$.
\end{Theorem}
Thus, the classical memory cost for storing the function $\chi$ (or a sufficiently close approximation to it)  can be high only if the contextual fraction of the equivalent MBQC substantially deviates from zero. Furthermore, by comparison of Eqs.~(\ref{Mavio}) and (\ref{MemCo}), we find that the upper bounds on the violation  $\Delta_\chi(\rho)$  of non-contextuality inequalities and on the information $I$ depend on the quantum state $\rho$ and the function $\chi$ only through the product ${\sf{CF}}(\rho)\H(\chi,\overline{\Lambda})$.\medskip 

With extra conditions on the structure of the set $E_0$, e.g. through the invariance of $E_0$ under $Q$, Theorem~\ref{T2} can be extended to bound the operational cost of evaluating the function $\chi$; see Appendix~\ref{CC}.
\medskip

{\em{Proof of Theorem~\ref{T2}.}} We prove the statement by explicitly constructing an algorithm that computes $\chi$ and satisfies the conditions of the theorem. We start with a whole family of algorithms to compute $\chi$, and later pick one member. These algorithms use the best ncHVM approximation $\mathfrak{s}_\text{opt}\in \overline{\Lambda}$ of $\chi$ and a list $L$ of exceptions. Any list $L$ is a subset $L \subset L_\text{max}$, where
$$
L_\text{max} = \{\left(a,\chi(a)-\mathfrak{s}_\text{opt}(a)\right)|\; a\in E_0,\, \chi(a)\neq \mathfrak{s}_\text{opt}(a) \}.
$$
The algorithms are as follows: Given an input $a$, if $(a,\delta(a)) \in L$ for some $\delta(a)$ then the output is $\chi(a) = \mathfrak{s}_\text{opt}(a) + \delta(a)$, and otherwise the output is $\chi(a)=\mathfrak{s}_\text{opt}(a)$.
 
Within this family of classical algorithms for computing $\chi$, we choose a list $L$ of exceptions such that $|L| = \left\lceil{\sf{CF}}(\rho) \H(\chi,\overline{\Lambda}) \right\rceil$. The resulting function evaluations thus fails for $\lfloor{(1-\sf{CF}}(\rho))\H(\chi,\overline{\Lambda})\rfloor$ of the $|E_0|$ inputs, and the average success probability of function evaluation therefore is
$$
\overline{p}_S = 1 - \frac{\lfloor{(1-\sf{CF}}(\rho))\H(\chi,\overline{\Lambda})\rfloor}{|E_0|}.
$$
This equals (or slightly exceeds by virtue of rounding) the upper limit of what the MBQC with contextual fraction ${\sf{CF}}(\rho)$ can reach, cf. Theorem~\ref{T_frac_beta}. The algorithm is thus correct.

To recover the function $\chi$ with sufficient accuracy, the optimal value assignment $\mathfrak{s}$ and the list $L$ of exceptions are stored. The memory cost of storing the list $L$, with its $|L| = \left\lceil{\sf{CF}}(\rho) \H(\chi,\overline{\Lambda}) \right\rceil$ items, is $(\lceil \log_2|E_0|\rceil  + \lceil \log_2d-1\rceil)  \lceil {\sf{CF}}(\rho)\, \H(\chi,\overline{\Lambda})\rceil $. 
The memory cost for storing $\mathfrak{s}_\text{opt}$ is as follows.
With the special choice $\eta_0$ for the function $\eta$ it holds that $\beta\equiv 0$, and Eq.~(\ref{beta_ds}) implies that $d\mathfrak{s} =0$. Hence, $\overline{\Lambda}$ is a vector space, of rank $\log_d \left|\overline{\Lambda}\right|$. Therefore, the function $\mathfrak{s}_\text{opt} \in \overline{\Lambda}$ is fully specified by $\log_d\left| \overline{\Lambda} \right|$ evaluations of $\mathfrak{s}_\text{opt}$. The cost of storing this information is $\lceil \log_2d\rceil \log_d\left|\overline{\Lambda}\right|$ bits.  

Adding these two contributions gives the r.h.s. of (\ref{MemCo}). The minimal memory cost is the same or lower.  $\Box$\smallskip

We note that contextuality can also place {\em{lower}} bonds on the memory requirements for classically simulating quantum phenomena \cite{Kar}.

\section{Conclusion}\label{Concl}

In this paper, we have provided state-dependent probabilistic contextuality proofs in which the resource-theoretic perspective on quantum contextuality and the cohomological perspective  are combined.  The resource perspective is important because of the recently discovered connection between contextuality and quantum computation \cite{AB}, \cite{How}.The cohomological perspective finds strong relevance in MBQC, since even the simplest example of a contextual MBQC \cite{AB} has cohomological interpretation \cite{Coho}.

Furthermore, we have advanced the cohomological viewpoint to probabilistic state-dependent contextuality proofs. These proofs are based on contextuality witnesses, i.e., expectation values of suitable linear operators. Contextuality is demonstrated whenever the value of a witness exceeds a corresponding threshold. The cohomological aspect of this is that the threshold value is a cohomological invariant; cf. Theorems~\ref{Hamming}, \ref{Hamming2}.
	
	We have also unified the cohomological perspective with the resource perspective. At the center of this unification stands the notion of the contextual fraction \cite{ABsheaf}. We have provided the following results involving it: 
\begin{itemize}
	\item{The maximum possible amount of violation of cohomological non-contextuality inequalities is proportional to the contextual fraction of the considered setting; see Eq.~(\ref{Mavio}).}
	\item{The contextual  fraction has an operational interpretation that links it to classical computation. Namely, the classical evaluation of a Boolean function can be hard only if the MBQC evaluation of the same function requires a large contextual fraction; see Theorems~\ref{T2} and \ref{T2b}.}
\end{itemize}
At first sight, the cohomological language may seem a complication, but the opposite is the case. The cohomological viewpoint removes decorum and reveals the essential and invariant features of parity-based and symmetry-based contextuality proofs.

\medskip

{\em{Acknowledgements.}} This work is funded by NSERC (CO, RR), the Stewart Blusson Quantum Matter Institute (CO), and Cifar (RR).   

\appendix

\section{The contextual fraction bounds the cost of function evaluation}\label{CC}

With additional assumptions on the structure of the set $E_0$ that hold for measurement-based quantum computation, we can extend Theorem~\ref{T2} to a relation between the contextual fraction of an MBQC and the operational cost of classical function evaluation. 
 
We consider the l2-MBQC; see \cite{RR13} or \cite{Abram3} for the full definition. This MBQC-variant formalizes the original scheme \cite{RB01}, and is characterized by two properties: (i) there is a choice between two measurement bases per local system, and (ii) the classical side-processing is mod 2 linear. We have the following result \cite{Abram3}, specialized to a single bit of output.
\begin{Theorem}[{\em{\cite{Abram3}}}]\label{T1}
Let $f: (\mathbb{Z}_2)^m \longrightarrow \mathbb{Z}_2$ be a Boolean function, and $\H(f,{\cal{L}})$ its Hamming distance to the closest linear function. For each l2-MBQC with contextual fraction ${\sf{CF}}(\rho)$ that computes $f$ with average success probability $\overline{p}_S$ over all $2^m$ possible inputs it holds that 
\begin{equation}\label{pSbd}
\overline{p}_S\leq 1- \frac{(1-{\sf{CF}}(\rho))\, \H(f,{\cal{L}})}{2^m}.
\end{equation}
\end{Theorem}
This result is a counterpart to Theorem~\ref{T_frac_beta}   with $f=\chi$, adjusted to MBQC. It is instructive to first look at two limiting cases of Theorem~\ref{T1}. For ${\sf{CF}}(\rho)=1$, i.e., strong contextuality, it holds that $\overline{p}_S\leq 1$, and  the theorem is not constraining. For the opposite limit of a non-contextual hidden variable model, ${\sf{CF}}(\rho)=0$, the bound in Theorem~\ref{T1} reduces to $\overline{p}_S \leq 1- \H(f,{\cal{L}})/2^m$, which is the result of \cite{RR13}. 

Now in general, for a given non-linear function $f$,  the larger the contextual fraction ${\sf{CF}}(\rho)$, the higher the potentially reachable success probability of function evaluation. In this sense, the contextual fraction is an indicator of computational power of MBQC.\medskip

The evaluation of Boolean functions by classical means and via MBQC are related as follows.
\begin{Theorem}\label{T2b}
Consider an l2-MBQC with contextual fraction ${\sf{CF}}(\rho)$, probabilistically evaluating a Boolean function $f: (\mathbb{Z}_2)^m \longrightarrow \mathbb{Z}_2$ that has a Hamming distance $\H(f,{\cal{L}})$ to the set of linear functions. If the closest linear function $g$ to $f$ is known, then  the operational cost $C_\text{op}$ of classically computing $f$ with at least the same probability of success are bounded by
$$
C_\text{op} \leq O\left(m\, \log_2{\sf{CF}}(\rho)\, \H(f,{\cal{L}})\right).
$$
\end{Theorem}
Thus, the evaluation of a given function with a target probability of success can be a hard task for classical computers only if the contextual fraction of the equivalent MBQC substantially deviates from zero.

As for the classical computational model whose performance is compared to the MBQC, we consider a dedicated device hard-wired to compute $f$. The MBQC itself---with fixed resource state and measurement sequence---is a hard-wired device too, and thus the comparison is fair. Using a dedicated device to classically compute the function $f$ justifies the assumption of Theorem~\ref{T2b} that the best linear approximation $g$ to $f$ is known.

Theorem~\ref{T2b} is a counterpart to similar results invoking entanglement \cite{Vidal}, \cite{MVdN} or the negativity of Wigner functions and similar quasi-probability distributions \cite{Pasha}---some applying to MBQC and others to the circuit model and quantum computation with magic states. For reference, we quote here  a result on the role of entanglement in MBQC\footnote{Theorem~\ref{T3} as stated here is a combination of Theorems 4 and 6 in \cite{MVdN}. Their Theorem~4 is broader in that it does not only refer to graph states but all quantum states of a fixed number of spins. However, it also comes with additional conditions concerning the knowledge of the optimal tensor network decomposition of the state. For graph states, these extra conditions can be eliminated, cf. Theorem~6 in \cite{MVdN}.} \cite{MVdN},
\begin{Theorem}[{\em{\cite{MVdN}}}]\label{T3}
Let $|G\rangle$ be an $n$-party graph state, and be $\tau$ the entanglement rank width of $|G\rangle$.  Then, any MBQC on $|G\rangle$ can be simulated classically in $O(n\,poly(2^\tau))$ time.
\end{Theorem}
Therein, the entanglement rank width is a proper entanglement monotone \cite{MVdN}. MBQC can solve a hard computational problem only if the entanglement in the resource state---as measured by the specific monotone of rank width---is substantial.

The structural likeness of Theorem~\ref{T2b} and Theorem~\ref{T3} is apparent, and, in fact, the same structure is present in the other results mentioned: All these theorems state an upper bound on the classical computational cost of reproducing the output of the quantum computation; and this upper bound is a monotonically increasing function in some measure of quantumness.

But there is also a difference. Theorem~\ref{T3} and the other results mentioned compete with the quantum protocol by simulating it classically. Theorem~\ref{T2b} admits further generality.  In this setting, we merely require of the classical algorithm that it evaluates the same function $f$ with the same average success probability. The theorem is agnostic about whether the classical algorithm achieves this by simulating the quantum protocol or by other means. \smallskip

The proof of Theorem~\ref{T2b} is very similar to the proof of Theorem~\ref{T2}.

{\em{Proof of Theorem~\ref{T2b}.}} We prove the statement by explicitly constructing an algorithm that computes $f$ and satisfies the conditions of the theorem. We consider family of algorithms to compute $f$ which use the best linear approximation $g$ of $f$ and a list $L$ of exceptions. Any list $L$ is such that $x\in L$ only if $f(x) \neq g(x)$, and otherwise the size $|L|$ of $L$ is a free parameter. The algorithms are as follows: Given an input $\textbf{i}$, if $\textbf{i} \in L$ then the output is $o= g(\textbf{i}) \oplus 1$, and otherwise the output is $o=g(\textbf{i})$. 
 
Within this family of classical algorithms for computing $f$, we choose a list $L$ of exceptions such that $|L| = \lceil{\sf{CF}}(\rho) \H(f,{\cal{L}}) \rceil$. The resulting function evaluations thus fails for $\lfloor{(1-\sf{CF}}(\rho))\H(f,{\cal{L}})\rfloor$ of the $2^m$ inputs, and the average success probability of function evaluation therefore is
$$
\overline{p}_S = 1 - \frac{\lfloor{(1-\sf{CF}}(\rho))\H(f,{\cal{L}})\rfloor}{2^m}.
$$
This equals (or slightly exceeds by virtue of rounding) the upper limit of what the MBQC with contextual fraction ${\sf{CF}(\rho)}$ can reach, cf. Theorem~\ref{T1}. The algorithm is thus correct.

The algorithm requires to evaluate the function $g$ on an input $\textbf{i}$, which takes $2m$ binary additions and multiplications, the lookup of the input $\textbf{i}$ in the list $L$, which takes $O(m\log_2|L|)$ operations, and the preparation of the output, which takes a constant number of operations. The operational cost is thus dominated by the lookup of the input $\textbf{i}$ in the list $L$, $C_\text{op}=O(m \log_2 {\sf{CF}(\rho)} \H(f,{\cal{L}}))$. The cost of the optimal algorithm to compute $f$ is the same or less. $\Box$\

\end{document}